\documentclass[11pt]{article}

\usepackage{amsmath,amssymb,amsthm}
\usepackage{tikz}
\usetikzlibrary{arrows}
\usepackage{palatino}
\usepackage{mathpazo}
\usepackage{stmaryrd}
\usepackage[margin=1in]{geometry}
\usepackage{mathtools}
\mathtoolsset{centercolon}
\usepackage{thm-restate}
\usepackage{hyperref}
\definecolor{darkred}  {rgb}{0.5,0,0}
\definecolor{darkblue} {rgb}{0,0,0.5}
\definecolor{darkgreen}{rgb}{0,0.5,0}
\hypersetup{
  pdftitle = {Asymptotic State Discrimination and a Strict Hierarchy in Distinguishability Norms},
  pdfauthor = {Eric Chitambar, Min-Hsiu Hsieh},
  colorlinks = true,
  urlcolor  = blue,         
  linkcolor = darkblue,     
  citecolor = darkgreen,    
  filecolor = darkred       
}

\newtheorem{theorem}{Theorem}
\newtheorem{proposition}{Proposition}
\newtheorem{lemma}{Lemma}

\theoremstyle{definition}

\newcommand{\ket}[1]{|#1\rangle}

\newcommand{\ip}[2]{\langle #1|#2\rangle}
\newcommand{\op}[2]{|#1\rangle \langle #2|}

\newcommand{\mc}[1]{\mathcal{#1}}
\newcommand{\mbf}[1]{\mathbf{#1}}

\newcommand{\aLOCC}{\overline{\text{LOCC}}}

\title{Asymptotic State Discrimination and a Strict Hierarchy in Distinguishability Norms}
\author{Eric Chitambar $^1$, Min-Hsiu Hsieh $^2$
\\[4mm]
\textit{$^1$ Department of Physics and Astronomy, Southern Illinois University,}\\ 
\textit{Carbondale, Illinois 62901, USA}\\
\textit{$^2$ Centre for Quantum Computation \& Intelligent Systems (QCIS),}\\
\textit{Faculty of Engineering and Information Technology (FEIT),}\\
\textit{University of Technology Sydney (UTS), NSW 2007, Australia}}

\date{}

\begin{document}
\maketitle

\begin{abstract}
In this paper, we consider the problem of discriminating quantum states by local operations and classical communication (LOCC) when an arbitrarily small amount of error is permitted.  This paradigm is known as asymptotic state discrimination, and we derive necessary conditions for when two multipartite states of any size can be discriminated perfectly by asymptotic LOCC.  We use this new criterion to prove a gap in the LOCC and separable distinguishability norms.  We then turn to the operational advantage of using two-way classical communication over one-way communication in LOCC processing.  With a simple two-qubit product state ensemble, we demonstrate a strict majorization of the two-way LOCC norm over the one-way norm.
\end{abstract}

\section{Introduction}

Any realistic scheme for processing information will inevitably encounter some experimental error.  Consider, for instance, the task of quantum state identification.  In the binary one-shot problem, a quantum system is prepared in some state $\rho$ with probability $p_\rho$ and another state $\sigma$ with probability $p_\sigma=1-p_\rho$.  A measurement is performed on the system, and based on this result, a guess is made on the state's identity.   The goal is to choose a measurement strategy that optimizes the probability of correctly identifying the state.  In practice, each experimental setup implementing this process will have unavoidable imperfections that generate a nonzero probability of error.  Hence for state discrimination, the experimental optimum refers to the success probability that can be approached arbitrarily close as the experimental errors are made smaller and smaller.  This paradigm is known as \textit{asymptotic state discrimination} since it involves a sequence of different measurement strategies with respective success probabilities that approach optimality in the limit.  

Asymptotic state discrimination is usually not considered on its own since for measurements performed across some quantum system $\mathcal{H}$, there is no effective difference between asymptotic and non-asymptotic processes.  To decide whether or not some success probability $p$ is asymptotically achievable, one need only consider whether it is theoretically possible to obtain $p$ \textit{exactly}.  This statement just reflects the mathematical fact that for some fixed number of outcomes, the set of positive-operator value measures (POVMs) on $\mathcal{H}$ is compact, and so the limit to any sequence of quantum measurements is itself a quantum measurement.  

The situation is quite different when the quantum system consists of $N$ different parts, $\mathcal{H}:=d_1\otimes...\otimes d_N$, and the subsystems are split between $N$ spatially separated parties.  In the state discrimination problem, the parties attempt to identify some globally shared state.  Being confined to their own respective laboratories, they can only investigate the state's identity by performing local measurements and globally communicating the measurement outcomes, a process known as local operations and classical communication (LOCC).  
Recently, much attention has been given to the mathematical properties of LOCC operations \cite{Kleinmann-2011a, Childs-2012a, Chitambar-2012c, Childs-2013a}, and unlike the full set of global measurements on $\mathcal{H}$, it turns out that the set of LOCC measurements is \textit{not} compact \cite{Chitambar-2012a, Chitambar-2012c}.  As a result, even if quantum theory prohibits LOCC from identifying two states with \textit{exactly} some success probability $p$, it still may be possible to attain  $p$ in the asymptotic sense.  Consequently, to truly characterize the limitations of realistic LOCC state discrimination, one must look beyond the question of exact discrimination and consider the asymptotic problem on its own.

To be a bit more formal, we say that a success probability $p$ can be obtained by asymptotic LOCC if for every $\epsilon>0$, there exists an LOCC protocol that correctly identifies the given state with probability at least $p-\epsilon$.  The POVM actually obtaining success probability $p$ belongs to $\aLOCC$, which is the set-closure of all LOCC POVMs \cite{Chitambar-2012c}.  To study $\aLOCC$, it is often helpful to consider the class of \textit{separable} operations (SEP).  A POVM $\{\Pi_i\}$ is $N$-partite separable if each $\Pi_i$ is separable, i.e. it can be expressed as a positive combination of product projectors onto $\mathcal{H}$, i.e. $\Pi_i=\sum_j\lambda_j\bigotimes_{k=1}^N\op{\phi_j^{(k)}}{\phi_j^{(k)}}$ with $\lambda_j>0$.  Any task achievable by $\aLOCC$ can also be achieved by SEP, however the converse is not true \cite{Bennett-1999a}.  Thus, to fully understand the operational power of $\aLOCC$, more fine-grained tools are needed beyond just a relaxation to separable operations.  


The class of asymptotic LOCC is also vital to the subject of distinguishability norms \cite{Matthews-2009a, Reeb-2011a, Lancien-2013a, Aubrun-2013a}.  Let $Herm(\mc{H})$ be the set of hermitian operators acting on $\mc{H}$.  Each $M\in Herm(\mc{H})$ can be uniquely decomposed as $M=R-S$ where $R$ and $S$ are orthogonal non-negative operators, and the trace norm of $M$ is given by $||M||_1=tr(R)+tr(S)$.  Denote a generic two-outcome ordered POVM by $\pi:=\{\Pi_R,\Pi_S\}$, and let $\Omega$ be any collection of such POVMs.  One can define the real function $||\cdot||_\Omega$ on $Herm(\mc{H})$ by 
\[||M||_{\Omega}:=\sup_{\pi\in\Omega}[tr(\Pi_R R)+tr(\Pi_S S)].\]  
A clear operational interpretation can be given to this function as follows.  For any ensemble of states $\{(p_\rho,\rho);(p_\sigma,\sigma)\}$, one defines the operator $M=p_\rho\rho-p_\sigma\sigma$.  Then $1/2(1-||M||_{\Omega})$ gives the infimum error probability of distinguishing $\rho$ from $\sigma$ among all POVMs belonging to $\Omega$ \cite{Matthews-2009a}.  For most operationally interesting classes of POVMs, it can be shown that $||\cdot||_{\Omega}$ is actually a norm on $Herm(\mathcal{H})$.  The well-known Helstrom bound is recovered when $\Omega$ is the full set of POVMs acting globally on $\mc{H}$; i.e. $||M||_{\text{GLOBAL}}=||M||_1$ \cite{Holevo-1973a, Helstrom-1976a}.  On the other hand, by choosing $\Omega$ to be the set of all LOCC POVMs, one obtains the so-called LOCC norm.  Here, however, non-compactness of LOCC means the error probability $1/2(1-||M||_{\text{LOCC}})$ is obtained by some POVM in $\aLOCC$ and not necessarily LOCC itself.


Distinguishability norms have an important connection to the phenomenon of data hiding \cite{Terhal-2001a, DiVincenzo-2002a, Matthews-2009a,  Lancien-2013a, Aubrun-2013a}.  For example, in any $d\otimes d$ system, pairs of orthogonal states $\rho$ and $\sigma$ exist for which $\tfrac{1}{2}||\rho-\sigma||_1=1$, and yet $\tfrac{1}{2}||\rho-\sigma||_{\text{LOCC}}\leq\frac{1}{d+1}$.  This implies the ability perfectly encode a single bit in the ensemble $\{(1/2,\rho);(1/2,\sigma)\}$ so that only an arbitrarily small amount of the information can be recovered using LOCC.  Distinguishability norms have also emerged as useful concepts in other contexts such as proving faithfulness of the squashed entanglement \cite{Brandao-2011a} and analyzing the quantum complexity in deciding separability \cite{Milner-2013a}.

With distinguishability norms, we have a mathematically precise way to compare different classes of POVMs.  We say that a class $\mathcal{C}$ is more powerful than class $\mathcal{C}'$ if $||M||_\mathcal{C}\geq ||M||_{\mathcal{C}'}$ for all $M\in Herm(\mc{H})$, and there exists some $M_0$ for which $||M_0||_\mathcal{C}> ||M_0||_{\mathcal{C}'}$.  In this paper, the classes of operations we will focus on are GLOBAL, SEP, LOCC, and 1-LOCC, where the latter refers to LOCC operations under the restriction of one-way classical communication.  In terms of distinguishability norms, these satisfy the ordering
\begin{equation}
\label{Eq:NormOrder}
||M||_1\geq ||M||_{\text{SEP}}\geq ||M||_{\text{LOCC}}\geq ||M||_{1-\text{LOCC}}.
\end{equation}
The data hiding literature offers nice examples of when the first inequality is strict \cite{DiVincenzo-2002a, Matthews-2008b}.  At the present, it is well-known that certain tasks can be accomplished by separable operations but not by LOCC \cite{Bennett-1999a, Duan-2007a, Chitambar-2012a}, and likewise, two-way LOCC offers greater possibilities than one-way LOCC \cite{Bennett-1996a, Cohen-2007a, Koashi-2007a, Owari-2008a, Xin-2008a, Chitambar-2013a, Nathanson-2013a}.  However, to our knowledge neither the second nor third inequalites in Eq. \ref{Eq:NormOrder} have been shown as strict.  Below we will construct explicit operators for which a separation between the distinguishability norms can be observed.

In general, it is not well-understood when bi-directional communication offers an operational advantage over uni-directional communication in LOCC information processing.  For instance, to convert one bipartite pure entangled state to another using LOCC, one-way classical communication suffices \cite{Lo-1997a}.  Another example even more relevant to the problem at hand involves the minimum error discrimination of any two bipartite pure states.  For pure states $\ket{\psi}$ and $\ket{\phi}$ of any dimension and any number of parties, the global optimal minimum-error discrimination can be achieved by one-way LOCC \cite{Walgate-2000a, Virmani-2001a}.  Thus, with respect to distinguishability norms, we have that
\begin{equation}
||M||_1=||M||_{\text{LOCC}}=||M||_{1-\text{LOCC}}
\end{equation}
whenever $M$ is rank two.  Given what we have learned in over 20 years of research into LOCC state discrimination, it is not surprising that the first equality fails to hold in general when $M$ has greater rank (one conclusive proof for the rank three case is given in \cite{Chitambar-2013b}).  Below we will show that likewise the second inequality fails to hold in general when $M$ is rank three.  This will be done by considering one pure state $\psi$ and one rank-two mixed state $\rho$ such that
\begin{equation}
||\psi-\rho||_{\text{LOCC}}>||\psi-\rho||_{1-\text{LOCC}}.
\end{equation}

  
\section{Conditions for Asymptotic Discrimination}

We begin by deriving a general necessary criterion for when two orthogonal states can be perfectly distinguished by $\aLOCC$.  Note that no generality is lost by restricting attention to perfect discrimination of orthogonal states since if $\rho$ and $\sigma$ are non-orthogonal, optimal discrimination in the minimum error sense is equivalent to asymptotic perfect discrimination of the two orthogonal (unnormalized) states $R$ and $S$, where $p_\rho\rho-p_\sigma\sigma=R-S$ \cite{Virmani-2001a, Chitambar-2013a}.

Our approach here is largely inspired by the work of Kleinmann \textit{et al.} \cite{Kleinmann-2011a}, and we utilize some of the tools introduced in that paper.  Let $\{(p_\rho,\rho);(p_\sigma,\sigma)\}$ be an ensemble of orthogonal states with any choice of prior probabilities satisfying $p_\sigma >p_\rho$.  We are free to choose $p_\sigma >p_\rho$ since the prior probabilities are irrelevant for perfect discrimination \footnote{Suppose $\sigma$ and $\rho$ can be perfectly distinguished by asymptotic LOCC with nonzero \textit{a priori} probabilities $(p_\sigma,p_\rho)$, and let $(p'_\sigma,p'_\rho)$ be any other \textit{a priori} probabilities.  For any $\epsilon>0$ there exists an LOCC POVM $\{\Pi^{(\epsilon)}_\sigma,\Pi^{(\epsilon)}_\rho\}$ such that $1-\epsilon<p_\sigma tr[\sigma \Pi^{(\epsilon)}_\sigma]+p_\rho tr[\rho\Pi^{(\epsilon)}_\rho]$.  From this we deduce $tr[\sigma \Pi^{(\epsilon)}_\sigma]>1-\epsilon/p_\sigma$ and $tr[\rho\Pi^{(\epsilon)}_\rho]>1-\epsilon/p_\rho$.  Hence, setting $\epsilon'=\epsilon\cdot( p_\sigma'/p_\sigma+p_\rho'/p_\rho)$ gives $p'_\sigma tr[\sigma \Pi^{(\epsilon)}_\sigma]+p'_\rho tr[\rho\Pi^{(\epsilon)}_\rho]>1-\epsilon'$.}.  Let $\mathcal{P}^{(\epsilon)}$ be any finite-round, finite-outcome LOCC protocol that errs in distinguishing $p_\sigma$ and $p_\rho$ with some probability $\leq\epsilon$.  As usual, we can envision protocol $\mathcal{P}^{(\epsilon)}$ as a tree where the \textit{root node} represents the first local measurement performed.  The different outcomes establish different branches which successively split into more branches after additional rounds of measurements.  
At any point in the protocol, we will refer to the \textit{bias} as the updated state probabilities given all the previous measurement outcomes.  We thus represent the bias by a two-component vector $(p_{\sigma|\lambda},p_{\rho|\lambda})$ where $p_{\sigma|\lambda}+p_{\rho|\lambda}=1$ and $\lambda$ indicates some particular sequence of outcomes.  Let $\Pi_\lambda$ be the POVM element corresponding to outcome sequence $\lambda$, and note that $\Pi_\lambda$ is a product operator.  Finally, the node in the protocol tree associated with $\lambda$ will be denoted by $\mbf{n}_\lambda$, and the probability of reaching this node is $p_\lambda$.  

Suppose that $\{\mbf{n}_{\lambda_k}\}_{k=0}^t$ are the various subsequent nodes obtained by performing a local POVM at node $\mbf{n}_\lambda$.  Then we say that $\mbf{n}_{\lambda_k}$ is a \textit{bias-flipped node} if $p_{X|\lambda}>p_{Y|\lambda}$ and $p_{Y|\lambda_k}\geq p_{X|\lambda_k}$ with $\{X,Y\}=\{\rho,\sigma\}$.  We say a node is a \textit{guessing node} if no more measurements are performed after reaching that node; at which point, the parties make a guess as to the state's identity.  If $\mbf{n}_\lambda$ is a guessing node, then the total error associated with the node is 
\[P_{err}(\mbf{n}_\lambda)=p_\lambda\cdot\min\{p_{\sigma|\lambda},p_{\rho|\lambda}\}.\]  
If $\mbf{n}_\lambda$ is a non-guessing node, the minimum guessing error attainable from node $\mbf{n}_\lambda$ onward is bounded by the global optimal probability $1/2(1-||M||_1)$.  At node $\mbf{n}_\lambda$, the \textit{a posteriori} ensemble is $\{(p_{\rho|\lambda},\rho_\lambda);(p_{\sigma|\lambda},\sigma_\lambda)\}$.  The trace norm can be related to the Hilbert-Schmidt inner product according to $||\alpha X-(1-\alpha)Y||_1\leq\sqrt{1-4\alpha(1-\alpha) tr[XY]}$ \cite{Nielsen-2000a}, and we thus obtain the following bound for non-guessing nodes:
\begin{align}
\label{Eq:PerrTrace}
P_{err}(\mbf{n}_\lambda)&\geq \tfrac{1}{2}[p_\lambda-\sqrt{p_\lambda^2-4 p_\rho p_\sigma tr(\Pi_\lambda\rho\Pi_\lambda\sigma)}].
\end{align}


With the above terminology fixed, our first task is to modify $\mathcal{P}^{(\epsilon)}$ so that the updated \textit{a posteriori} probabilities at each bias-flipped node are exactly equal.  There are various ways to perform this modification such as decomposing each strong measurement into a sequence of weak measurement so that the bias evolves by arbitrarily small increments throughout the protocol \cite{Bennett-1999a, Oreshkov-2005a}.  Alternatively, Kleinmann \textit{et al.} describe a ``pseudo-weak'' approach that simply breaks any measurement into two steps without affecting the overall success probability of the protocol \cite{Kleinmann-2011a}.  Either way, the following construction can be achieved.
\begin{proposition}
\label{Prop:LOCCweak}
Suppose that $\mathcal{P}^{(\epsilon)}$ is some LOCC protocol that distinguishes $\rho$ and $\sigma$ with error probability $\leq\epsilon$.  Then there exits another protocol $\hat{\mathcal{P}}^{(\epsilon)}$ that also distinguishes with error probability $\leq\epsilon$ but with all bias-flipped nodes having a bias $(1/2,1/2)$.  
\end{proposition}
Consider such a protocol $\hat{\mathcal{P}}^{(\epsilon)}$.  Starting from the root node where $p_\sigma>p_\rho$, we track the bias throughout the protocol.  Every branch will either reach a guessing node without the bias flipping at least once (say these nodes belong to set $\mc{B}_1$), or the branch will reach a node with bias $(1/2,1/2)$ (say these nodes belong to set $\mc{B}_2$).  Since $\min\{p_{\sigma|\lambda},p_{\rho|\lambda}\}=p_{\rho|\lambda}$ for all $\mbf{n}_\lambda\in\mc{B}_1$, we have that
\begin{align*}
\label{Eq:Err1}
P_{err}(\mathbf{n}_\lambda)&=p_\lambda p_{\rho|\lambda}=p_\rho p_{\lambda|\rho}=p_\rho tr(\Pi_\lambda\rho)\quad\forall \mbf{n}_\lambda\in\mc{B}_1.
\end{align*}
Summing over these nodes gives the bound
\begin{equation*}
\epsilon\geq\sum_{\lambda:\mbf{n}_\lambda\in\mc{B}_1}P_{err}(\mbf{n}_\lambda)=\sum_{\lambda:\mbf{n}_\lambda\in\mc{B}_1}p_\rho tr(\Pi_\lambda\rho)=p_\rho tr(\Pi^{(\epsilon)}\rho),
\end{equation*}
where $\Pi^{(\epsilon)}:=\sum_{\lambda:\mbf{n}_\lambda\in\mc{B}_1}\Pi_\lambda$.  For the nodes in $\mc{B}_2$, we use Eq. \eqref{Eq:PerrTrace} to obtain the bound
\begin{equation*}
\epsilon\geq\sum_{\lambda:\mbf{n}_\lambda\in\mc{B}_2}\tfrac{1}{2}[p_\lambda-\sqrt{p_\lambda^2-4p_\rho p_\sigma tr(\Pi_\lambda\rho\Pi_\lambda\sigma)}].
\end{equation*}
Thus, in order for an LOCC protocol to distinguish $\sigma$ and $\rho$ with error probability $\leq\epsilon$, there must exist a POVM $\{\Pi^{(\epsilon)},\Pi_\lambda\}_\lambda$ with $\Pi^{(\epsilon)}$ a separable operator and each $\Pi_\lambda$ a product operator, such that for all $\lambda$:
\begin{align}
p_\rho tr(\Pi^{(\epsilon)}\rho)&\leq\epsilon,\label{Eq:consE1}\\
\sum_{\lambda}\tfrac{1}{2}[p_\lambda-\sqrt{p_\lambda^2-4p_\rho p_\sigma tr(\Pi_\lambda\rho\Pi_\lambda\sigma)}]&\leq\epsilon,\label{Eq:consE2}\\
tr[\Pi_\lambda(p_\rho\rho-p_\sigma\sigma)]&=0,\label{Eq:consE3}
\end{align}
where $p_\lambda=p_\rho tr[\Pi_\lambda\rho]+p_\sigma tr[\Pi_\lambda \sigma]$.  Note that the error bounds contained in Eqns. \eqref{Eq:consE1}--\eqref{Eq:consE3} apply to an actual LOCC-implementable POVM.  

Unfortunately, the number of $\Pi_\lambda$ satisfying Eqns. \eqref{Eq:consE1}--\eqref{Eq:consE3} may be unbounded, and to deal with this, we slightly relax the LOCC-implementable condition.  This is done via Carath\'{e}odory's Theorem, which allows us to bound the number of $\Pi_\lambda$.  
\begin{lemma}[Carath\'{e}odory's Theorem~\cite{Rockafellar-1996a}]
Let $S$ be a subset of $\mathbb{R}^n$ and $conv(S)$ its convex hull.  Then any $x \in conv(S)$ can be expressed as a convex combination of at most $n+1$ elements of $S$.
\end{lemma}
\noindent  For any POVM $\{\Pi^{(\epsilon)},\Pi_\lambda\}_\lambda$ satisfying Eqns. \eqref{Eq:consE1}--\eqref{Eq:consE3}, define $S:=\{\tfrac{\Pi_\lambda}{tr[\Pi_\lambda]}\}_\lambda$ so that $\frac{1}{\eta}(\mathbb{I}-\Pi^{(\epsilon)})=\frac{1}{\eta}\sum_\lambda tr[\Pi_\lambda]\tfrac{\Pi_\lambda}{tr[\Pi_\lambda]}$ belongs to $conv(S)$, where $\eta=\sum_\lambda tr[\Pi_\lambda]\leq \prod_{i=1}^nd_i$.  The set of hermitian operators acting on $\mathcal{H}=\bigotimes_{i=1}^N d_i$ represents a $\prod_{i=1}^N d^2_i$-dimensional real vector space (every $d\times d$ hermitian matrix is specified by $d^2$ real numbers).  Thus by Carath\'{e}odory's Theorem, there exists a set of non-negative numbers $\{q_\lambda\}_{\lambda=1}^D$ with $D:=\prod_{i=1}^N d^2_i+1$ such that $\mathbb{I}-\Pi^{(\epsilon)}=\sum_\lambda q_\lambda\tfrac{\Pi_\lambda}{tr[\Pi_\lambda]}$ and $\sum_{\lambda=1}^Dq_\lambda=\eta$.

Let $Herm(\mathcal{H})^{\times (D+1)}$ denote the $(D+1)$-fold Cartesian product of hermitian operators acting on $\mathcal{H}$.  We construct a compact subset $\mathcal{Q}\subset\mathbb{R}^D\times Herm(\mathcal{H})^{\times (D+1)}$ as follows.  A collection $\{q_\lambda\}_{\lambda=1}^D\cup \{\Pi^{(\epsilon)},\Pi_\lambda\}_{\lambda=1}^D$ belongs to $\mathcal{Q}$ if: (i) $\Pi^{(\epsilon)}$ is non-negative separable and $\Pi_\lambda$ are non-negative product operators satisfying Eqns. \eqref{Eq:consE1}--\eqref{Eq:consE3} for some $\epsilon \in [0, p_\rho]$, (ii) the $q_\lambda$ are non-negative satisfying $0\leq\sum_\lambda q_\lambda\leq \prod_{i=1}^nd_i$, and (iii) $\mathbb{I}-\Pi^{(\epsilon)}=\sum_\lambda q_\lambda\tfrac{\Pi_\lambda}{tr[\Pi_\lambda]}$.  Under these conditions, $\mathcal{Q}$ is a closed, bounded, and therefore compact set.

By assumption of asymptotic discrimination, we must be able to find a collection $\{q_\lambda\}_{\lambda=1}^D\cup \{\Pi^{(\epsilon)},\Pi_\lambda\}_{\lambda=1}^D$ in $\mathcal{Q}$ for every $\epsilon>0$.  Compactness then assures the existence of some $\{q_\lambda\}_{\lambda=1}^D\cup\{\Pi^{(0)},\Pi_\lambda\}_{\lambda=1}^D$ satisfying Eqns. \eqref{Eq:consE1}--\eqref{Eq:consE3} for $\epsilon=0$, with $\mathbb{I}-\Pi^{(0)}=\sum_{\lambda=1}^D q_\lambda\tfrac{\Pi_\lambda}{tr[\Pi_\lambda]}$.  But with $\epsilon=0$, the elements $q_\lambda\tfrac{\Pi_\lambda}{tr[\Pi_\lambda]}$ themselves satisfy Eqns. \eqref{Eq:consE1}--\eqref{Eq:consE3}.  Note that the sum in \eqref{Eq:consE2} vanishes iff $tr(\Pi_\lambda\rho\Pi_\lambda\sigma)=0$ for each $\Pi_\lambda$.  Hence we obtain our main result:
\begin{theorem}
\label{thm:OurKKB}
If $N$-partite states $\rho$ and $\sigma$ can be perfectly distinguished by asymptotic LOCC, then for each $x\in [1/2,1]$ there must exist a POVM $\{\Pi_0,\Pi_\lambda\}_{\lambda=1}^D$ such that $\Pi_0$ is a separable operator, each $\Pi_\lambda$ is a product operator, and
\begin{align}
tr(\Pi_0\rho)=0&,\label{Eq:consE1F}\\
tr(\Pi_\lambda\rho\Pi_\lambda\sigma)=0&,\qquad\forall 1\leq \lambda\leq D\label{Eq:consE2F}\\
tr[\Pi_\lambda[(1-x)\rho-x\sigma]]=0&,\qquad\forall 1\leq\lambda\leq D\label{Eq:consE3F}
\end{align}  
where $D=\prod_{k=1}^N d_k^2 +1$ and $d_k$ is the dimension of system $k$.
\end{theorem}

In the above proof, we began by assuming that $p_\sigma>p_\rho$, and this corresponds to the choice $x\in(1/2,1)$.  The boundary point $x=1/2$ can be trivially satisfied by the identity.  For $x=1$,  note that distinguishability by $\aLOCC$ implies distinguishability by SEP.  Hence, there must exist two separable operators $\Pi_\rho$ and $\Pi_\sigma$ such that $tr[\Pi_\rho\sigma]=0$, $tr[\Pi_\sigma\rho]=0$, and $\Pi_\rho+\Pi_\sigma=\mathbb{I}$.  Setting $\Pi_\sigma=\Pi_0$ and decomposing $\Pi_\rho$ into a convex sum of product operators $\Pi_\lambda$ provides the necessary ingredients to satisfy Theorem \ref{thm:OurKKB} for $x=1$.

The case of $x=1$ allows us to immediately draw conclusions about states $\rho$ and $\sigma$ whose supports cover the full state space.  Namely, when $\mc{H}=supp(\rho)\oplus supp(\sigma)$, from the discussion of the previous paragraph it follows that distinguishability by $\aLOCC$ requires $supp(\rho)$ and $supp(\sigma)$ to each possess a product state basis.  We can apply this observation to the LOCC distinguishability norm of certain hermitian matrices having full rank.  Recall that $||M||_{1}=||M||_{\Omega}$ for some set of POVMs $\Omega$ iff the orthogonal ensemble $\{(p_\rho, \rho),(p_\sigma,\sigma)\}$ can be perfectly distinguished by some POVM in $\overline{\Omega}$, where $M=p_\rho \rho-p_\sigma\sigma$.  We refer to the positive (resp. negative) eigenspace of an operator as the subspace spanned by eigenvectors whose corresponding eigenvalues are positive (resp. negative).
\begin{lemma}
Let $M$ be a full rank $d_A\otimes d_B$ hermitian operator possessing either a positive or negative eigenspace of dimension two.  Then $||M||_1=||M||_{LOCC}$ iff an orthonormal product basis exists for both the positive and negative eigenspaces.
\end{lemma}

\begin{proof}
Let $M=R-S$ be the orthogonal decomposition of $M$ where $rk(R)=2$ and $R,S\geq 0$.  Assume that $R$ and $S$ can be perfectly distinguished by $\aLOCC$ with the POVM $\{\Pi_R,\Pi_S\}$.  From the above observation, we have that both $supp(R)$ and $supp(S)$ must contain a product state basis.  If $supp(R)$ is a tensor product space (i.e. of the form $\mathbb{C}^1\otimes \mathbb{C}^2$), then $supp(R)$ obviously has an orthonormal product basis as well as $supp(S)$.  Suppose then that $supp(R)$ is not a tensor product space.  Then since any such two-dimensional subspace can possess at most two product states, the product state basis of $supp(R)$ will be unique.  Now suppose that the product states in $supp(R)$ are not orthogonal and given by $\ket{0}\otimes\ket{0}$ and $(\alpha \ket{0}+\ket{1})\otimes(\beta\ket{0}+\ket{1})$ with $\alpha,\beta\not=0$.  Let $\Pi_0$ denote the projection onto the subspace spanned by $\{\ket{i}\otimes\ket{j}\}_{i,j=0}^1$.  Note that $\Pi_0 R\Pi_0=R$ and likewise $\Pi_0 \Pi_R\Pi_0=\Pi_R$ since $tr[\Pi_R R]=1$ and $tr[\Pi_R S]=0$.  This implies that the $2\otimes 2$ POVM $\{\Pi_R, \Pi_0\Pi_S\Pi_0\}$ perfectly distinguishes $R$ and $\Pi_0 S\Pi_0$.  The POVM $\{\Pi_R, \Pi_0\Pi_S\Pi_0\}$ belongs to SEP, and the conditions for distinguishing the rank-two elements $R$ and $\Pi_0 S \Pi_0$ by SEP is that $supp(R)$ contains an orthogonal product basis \cite{Chitambar-2013b}.  This is a contradiction.

On the other hand, suppose that $supp(R)$ (and therefore also $supp(S)$) contains an orthogonal product basis of the form $\ket{00}$ and $\ket{1 \beta}$ (without loss of generality). Then Alice measures in the computational basis.  If she obtains outcome $\ket{0}$, Bob projects into any orthogonal basis containing $\ket{0}$, and they choose state $R$ iff he obtains $\ket{0}$.  If she obtains outcome $\ket{1}$, Bob projects into any orthogonal basis containing $\ket{\beta}$ and they choose $R$ iff he obtains $\ket{\beta}$.
\end{proof}

\section{Separating SEP and LOCC Norms}

We will next apply Theorem \ref{thm:OurKKB} in a straightforward manner to show a gap between the SEP and LOCC norms.  To achieve this, we will return to the original paper that demonstrated the subtle difference between separability and locality in terms of distinguishability  \cite{Bennett-1999a}.  The authors presented nine orthogonal product states spanning the full $3\otimes 3$ state space:
\begin{align}
\ket{\psi_0}&=\ket{1}\otimes\ket{1},\notag\\
\ket{\psi_{1\pm}}&=\ket{0}\otimes\ket{0\pm 1},&\ket{\psi_{2\pm}}&=\ket{0\pm 1}\otimes\ket{2},\notag\\
\ket{\psi_{3\pm}}&=\ket{1\pm 2}\otimes\ket{0},&\ket{\psi_{4\pm}}&=\ket{2}\otimes\ket{1\pm 2}.
\end{align}
It was shown that as a nine-state ensemble, $\aLOCC$ is unable to perfectly identify a given state while SEP, in contrast, can achieve the task.  Here, we wish to prove the stronger result that the following mixtures cannot be perfectly distinguished by $\aLOCC$:
\begin{align}
\sigma&=\frac{1}{4}\sum_{i=1}^4\op{\psi_{i+}}{\psi_{i+}},\notag\\
\rho&=\frac{1}{5}(\op{\psi_0}{\psi_0}+\sum_{i=1}^4\op{\psi_{i-}}{\psi_{i-}}).
\end{align}

Let us show that these states cannot be discriminated by $\aLOCC$.  Eq. \eqref{Eq:consE2F} necessitates that each $\Pi_\lambda=A_\lambda\otimes B_\lambda$ acts invariantly on $supp(\sigma)$.  The key observation is that the $\ket{\psi_{i+}}$ are the unique product states lying in $supp(\sigma)$.  To see this, note that any product state in $supp(\sigma)$ can be represented as a matrix $\sum_{i=1}^4\alpha_{i}\op{\psi_{i+}}{\psi_{i+}}$ whose $2\times 2$ minors vanish; this requires all but one $\alpha_i$ to be zero.  As a result of this property, up to overall non-negative scalars, $A_\lambda\otimes B_\lambda$ must map the set of four states
\begin{align}
\ket{\psi_{1+}}&=\ket{0}\otimes \ket{0+1} ,&\ket{\psi_{2+}}&=\ket{0+1}\otimes \ket{2},\notag\\
\ket{\psi_{3+}}&=\ket{1+2}\otimes \ket{0},&\ket{\psi_{4+}}&=\ket{2}\otimes\ket{1+2}\notag
\end{align}
onto itself.  Now the action $M\ket{i}=c\ket{j+k}$ is not possible for any $M\geq 0$ when $i\not=j,k$ and $c\not=0$.  Applying this fact to both $A_\lambda$ and $B_\lambda$ implies that $A_\lambda\otimes B_\lambda\ket{\psi_{i+}}=c_{i+}\ket{\psi_{i+}}$ where $c_{i+}\geq 0$.  

Suppose now that there are two states $\ket{\psi_{i_1+}}=\ket{\alpha_{i_1}}\ket{\beta_{i_1}}$ and $\ket{\psi_{i_2+}}=\ket{\alpha_{i_2}}\ket{\beta_{i_2}}$ for which both $c_{i_1+},c_{i_2+}>0$.  We can always find another state $\ket{\psi_{i_3+}}=\ket{\alpha_{i_3}}\ket{\beta_{i_3}}$ with $i_3\not=i_1,i_2$ such that both $\ip{\alpha_{i_1}}{\alpha_{i_3}}\not=0$ and $\ip{\beta_{i_2}}{\beta_{i_3}}\not=0$.  But since the eigenspaces of $A_\lambda$ (resp. $B_\lambda$) are orthogonal, it follows that $A_\lambda\otimes B_\lambda\ket{\alpha_{i_3}}\ket{\beta_{i_3}}\not=0$; for if $\ket{\alpha_{i_3}}$ (resp. $\ket{\beta_{i_3}}$) were to be in the kernel of $A_\lambda$ (resp. $B_\lambda$), then a non-orthogonal state would also lie in the support of $A_\lambda$ (resp. $B_\lambda$).  Now with $c_{i_1+},c_{i_2+},c_{i_3+}>0$ and because the local parts of any three $\ket{\psi_{i+}}$ are linearly independent, $A_\lambda$ and $B_\lambda$ are necessarily full rank.  In this case, both $A_\lambda$ and $B_\lambda$ will have eigenstates $\{\ket{0},\ket{0+1},\ket{1+2},\ket{2}\}$, which, by a simple calculation, can be seen as possible only if both $A_\lambda$ and $B_\lambda$ are proportional to the identity.  

We have just shown that if $A_\lambda\otimes B_\lambda$ is not proportional to the identity, then $A_\lambda\otimes B_\lambda$ must eliminate at least three of the $\ket{\psi_{i+}}$.  However, by taking $x\in(1/2,1)$ in Eq. \eqref{Eq:consE3F}, two conditions are ensured: $A_\lambda\otimes B_\lambda$ is not proportional to the identity (since $x\not=1/2$), and $c_{i+}$ is nonzero for at least one value of $i$ (since $x\not=1$).  Hence, for $x$ in this interval, $A_\lambda\otimes B_\lambda$ must eliminate three and only three of the $\ket{\psi_{i+}}$.  By again using the fact that Alice and Bob's parts are linearly independent for any three of the $\ket{\psi_{i+}}$, this is possible only if $A_\lambda\otimes B_\lambda$ is rank two and having the form
\begin{equation}
\label{Eq:DistinguishPOVM}
A_\lambda\otimes B_\lambda=c_{i+}\op{\psi_{i+}}{\psi_{i+}}+c_{i-}\op{\psi_{i-}}{\psi_{i-}}
\end{equation}
with $c_{i+},c_{i-}>0$.  The fact that $c_{i-}>0$ again follows from Eq. \eqref{Eq:consE3F}.  In summary, each $\Pi_\lambda$ has support on a two-dimensional space spanned by $\{\ket{\psi_{i+}},\ket{\psi_{i-}}\}$ for some $i\in\{1,2,3,4\}$.  But then discrimination of $\sigma$ and $\rho$ is impossible by $\aLOCC$ since Eq. \eqref{Eq:consE1F} together with the fact that $\Pi_0+\sum_\lambda\Pi_\lambda=\mathbb{I}$ implies that $supp(\rho)\subset \bigcup_\lambda supp(\Pi_\lambda)$.  But $\ket{1}\otimes\ket{1}\in supp(\rho)$ will not be contained in $\bigcup_\lambda supp(\Pi_\lambda)$.

We have thus demonstrated a gap between the distinguishability norms:
\begin{equation}
||\rho-\sigma||_{\text{SEP}}>||\rho-\sigma||_{\text{LOCC}}.
\end{equation}
This relatively simple argument can be applied to more ensembles with the same type of structure (see Ref. \cite{DiVincenzo-2003a} for such ensembles).

The example constructed in this section demonstrates the difference between Theorem \ref{thm:OurKKB} and the distinguishability criterion of Ref. \cite{Kleinmann-2011a}.  The criterion given in Proposition 1 of Ref. \cite{Kleinmann-2011a} will not show the impossibility of discriminating $\rho$ and $\sigma$ by $\aLOCC$; indeed, a product operator of the form given in Eq. \eqref{Eq:DistinguishPOVM} will satisfy that distinguishability criterion.  The essential component of Theorem \ref{thm:OurKKB} used to eliminate the possibility of $\aLOCC$ discrimination is that the collective supports of $\Pi_\lambda$ must cover the full support of $\rho$.

\section{Separating One-Way and Two-Way LOCC Norms}
\label{Sect:1-wayVs2-way}

By the ``two-way'' LOCC norm, we are referring to the general LOCC norm.  To show a gap between $||\cdot||_{\text{LOCC}}$ and $||\cdot||_{1-\text{LOCC}}$,  the example ensemble we consider consists of the equiprobable two-qubit states
\begin{align}
\label{Eq:Koashi}
\psi&=\op{00}{00}&\rho&=1/2(\op{++}{++}+\op{--}{--}),
\end{align}
where $\ket{\pm}=\tfrac{1}{\sqrt{2}}(\ket{0}\pm\ket{1})$.  Distinguishability of this ensemble was analyzed by Koashi \textit{et al.} within the context of \textit{unambiguous} discrimination \cite{Koashi-2007a}.  However, for the task of minimum error discrimination, this simple-structure ensemble has only been in studied in Ref. \cite{Chitambar-2013b} where it was shown that separable operators are able to obtain the global minimum error probability, while this rate cannot be achieved in finite rounds of LOCC.  In terms of norms, we can summarize this result as
\begin{equation}
||\psi-\rho||_1=||\psi-\rho||_{\text{SEP}}>||\psi-\rho||_{r-\text{LOCC}},
\end{equation} 
where $r\in\mathbb{Z}_+$ and $r-\text{LOCC}$ denotes the set of $r$-round two-outcome LOCC POVMs (acting on two qubits).  We emphasize that despite the previous unambiguous discrimination analysis conducted by Koashi \textit{et al.} on ensemble \eqref{Eq:Koashi}, it is a completely different problem to consider the minimum error discrimination of these states.  Indeed, we have recently observed that SEP and LOCC are equally powerful for distinguishing certain states when optimal unambiguous discrimination is the figure of merit, but SEP and LOCC are inequivalent for the same states when minimum error discrimination is taken as the figure of merit \cite{Chitambar-2013a,Chitambar-2013b}.

The goal at hand is to prove that $||\psi-\rho||_{\text{LOCC}}>||\psi-\rho||_{1-\text{LOCC}}$.  To do so we will first compute the minimum one-way LOCC error probability; then we will construct a specific two-way protocol that beats it.  For both parts, we will need the explicit formula for Helstrom's bound when distinguishing one pure qubit state from one mixed state.  This is given by computing $||p_0\op{\psi_0}{\psi_0}-(p_1\op{\psi_1}{\psi_1}+p_2\op{\psi_2}{\psi_2})||_1$ for qubit pure states $\{\ket{\psi_i}\}_{i=0}^2$ and $\sum_{i=0}^2p_i=1$.  It is relatively straightforward to make this calculation (see Ref. \cite{Chitambar-2013a} for details); one finds the following.
\begin{lemma}
\label{Lem:pure-mixed}
Consider the weighted states $p_0\op{\psi_0}{\psi_0}$ and $p_1\op{\psi_1}{\psi_1}+p_2\op{\psi_2}{\psi_2}$.  The minimum error probability is
\begin{align}
\label{Eq:prob}
\tfrac{1}{2}-\tfrac{1}{2}\sqrt{|p_0-p_1-p_2|^2-4\det\Delta}
\end{align}
if $\det\Delta\leq 0$, and $\tfrac{1}{2}-\tfrac{1}{2}|p_0-p_1-p_2|$ if $\det\Delta\geq 0$, where  $\Delta:=p_0\op{\psi_0}{\psi_0}-p_1\op{\psi_1}{\psi_1}-p_2\op{\psi_2}{\psi_2}$ and 
\begin{equation}
\det \Delta=p_1p_2(1-|\ip{\psi_1}{\psi_2}|^2)-p_0p_1(1-|\ip{\psi_0}{\psi_1}|^2)-p_0p_2(1-|\ip{\psi_0}{\psi_2}|^2).
\end{equation}
\end{lemma}

\medskip

\noindent \textbf{One-Way Optimal:}  Consider a one-way measurement scheme that is optimal in the minimum error sense.  Without loss of generality, we consider communication from Alice to Bob with Alice's measurement consisting of rank-one POVMs.  The joint measurement is then given by $\{\op{a_\lambda}{a_\lambda}\otimes B_{\mu|\lambda}\}$ where $\{B_{\mu|\lambda}\}_{\mu=0}^1$ is a conditional POVM performed by Bob.  Note that we are dealing with a real ensemble with each state being $(\sigma_z\otimes \sigma_z)$-invariant.  Consequently, we can simplify the structure of Alice's POVM according to the following proposition.
\begin{proposition}
\label{Prop:AlicePOVM}
An optimal one-way LOCC scheme consists of Alice's POVM being
\[\{p_0\op{\phi_0}{\phi_0},p_0\sigma_z\op{\phi_0}{\phi_0}\sigma_z,p_1\op{\phi_1}{\phi_1},p_1\sigma_z\op{\phi_1}{\phi_1}\sigma_z,\}\]
where $\ket{\phi_i}=\cos\phi_i/2\ket{0}+\sin\phi_i/2\ket{1}$. 
\end{proposition}  
\begin{proof}
Appendix A
\end{proof}
With this proposition, it suffices to consider only two measurement outcomes corresponding to POVM elements $q_0\op{\phi_0}{\phi_0}$ and $q_1\op{\phi_1}{\phi_1}$ for which
\begin{align}
q_0+q_1&=1,&q_0\cos\phi_0+q_1\cos\phi_1&=0.
\end{align}
The full POVM will then contain the $\sigma_z\otimes\sigma_z$-rotated elements as well, and the total average error probability will be given by $2P(\lambda=0)P(Err|\lambda=0)+2P(\lambda=1)P(Err|\lambda=1)$.   Using the relations
\begin{align}
|\ip{\phi_\lambda}{0}|^2&=\tfrac{q_\lambda}{2}(1+\cos\phi_\lambda),&|\ip{\phi_\lambda}{\pm}|^2&=\tfrac{q_\lambda}{2}(1\pm\sin\phi_\lambda),\notag
\end{align}
the unnormalized states that Bob must distinguish given outcome $\lambda\in\{0,1\}$ are 
\[\frac{p_{\lambda|0}}{2P(\lambda)}\op{0}{0},\qquad \frac{p_{\lambda|+}}{4P(\lambda)}\op{+}{+}+\frac{p_{\lambda|-}}{4P(\lambda)}\op{-}{-},\]
where
\begin{align}
p_{\lambda|0}&=\tfrac{q_\lambda}{2}(1+\cos\phi_\lambda), &p_{\lambda|\pm}&=\tfrac{q_\lambda}{2}(1\pm\sin\phi_\lambda),&P(\lambda)&=\tfrac{q_\lambda}{2}(1+\tfrac{\cos\phi_\lambda}{2}).
\end{align}
For outcome $\lambda$, a direct calculation gives
\begin{align}
\label{Eq:values}
|\frac{p_{\lambda|0}}{2P(\lambda)}-\frac{p_{\lambda|+}}{4P(\lambda)}-\frac{p_{\lambda|-}}{4P(\lambda)}|&=\frac{q_\lambda}{4P(\lambda)}|\cos\phi_\lambda|\notag\\
\det \Delta&=\left(\frac{q_\lambda}{8 P(\lambda)}\right)^2[(1-\cos\phi_\lambda)^2-3].
\end{align}
Since the average error is given by $2\sum_{\lambda=0}^1P(\lambda)P(Err|\lambda)$, using Lemma \ref{Lem:pure-mixed}, the error for when both outcomes satisfy $\det\Delta\leq 0$ is given by
\begin{align}
&2\sum_{\lambda=0}^1\left(\frac{P(\lambda)}{2}-\frac{P(\lambda)}{2}\sqrt{\left(\frac{q_\lambda}{4P(\lambda)}|\cos\phi_\lambda|\right)^2-4\left(\frac{q_\lambda}{8 P(\lambda)}\right)^2[(1-2\cos\phi_\lambda)]}\right)\notag\\
\label{Eq:detleq0}
&=\frac{1}{2}\left(1-\frac{1}{\sqrt{2}}[q_0\sqrt{1+\cos\phi_0}+q_1\sqrt{1+\cos\phi_1}]\right).
\end{align}
We want to minimize this under the constraints that $q_0+q_1=1$ and $q_0\cos\phi_0+q_1\cos\phi_1=0$.  Using concavity of the function $\sqrt{1+x}$, Eq. \eqref{Eq:detleq0} immediately gives a lower bound of $\frac{1}{2}\left(1-\frac{1}{\sqrt{2}}\right)$.  In fact, this lower bound is saturated by the choice of $q_0=q_1=1/2$, $\phi_0=\pi/2$ and $\phi_1=3\pi/2$.  This corresponds to Alice performing the projective measurement $\{\op{+}{+},\op{-}{-}\}$.

The only other possibility is if $\det\Delta>0$ for outcome $\lambda=0$, and $\det\Delta\leq 0$ for outcome $\lambda=1$.  In this case, the average error is given by
\begin{align}
\label{Eq:detgeq0}
2\times \left(\frac{1}{4}-\frac{q_0}{8}|\cos\phi_0|-\frac{q_1}{4\sqrt{2}}\sqrt{1+\cos\phi_1}\right)&=\frac{1}{2}\left(1-q_1(\frac{\cos\phi_1}{2}+\frac{\sqrt{1+\cos\phi_1}}{\sqrt{2}})\right)
\end{align}
where we have used the relation $q_0\cos\phi_0+q_1\cos\phi_1=0$ and the fact that $\cos\phi_0<0$.  This is minimized using a Lagrange multiplier, and the calculation is carried out in Appendix B.  One finds that two extrema exists for Alice's measurement: when she measures in the computational basis $\{\ket{0},\ket{1}\}$ and when she measures in the Hadamard basis $\{\ket{+},\ket{-}\}$.  Measuring in the computational basis leads to a smaller error probability.  It corresponds to choosing $q_1=1/2$ and $\cos\phi_1=1$ so that the the optimal one-way LOCC error probability is $1/8$.

\medskip

\noindent\textbf{Two-Way Improvement:}

Now we construct an improved two-way protocol.  We will track the evolution of the three-state ensemble 
\[\{(\tfrac{1}{2},\ket{00}); (\tfrac{1}{4},\ket{++});(\tfrac{1}{4},\ket{--})\},\]
keeping in mind that the actual problem includes a mixing over the last two states.  Suppose that Alice performs a two-outcome measurement with Kraus operators given by
\begin{align}
A_0&=\sqrt{1/2(1+p)}\op{0}{0}+\sqrt{1/2(1-p)}\op{1}{1}\notag\\
A_1&=\sqrt{1/2(1+p)}\op{1}{1}+\sqrt{1/2(1-p)}\op{0}{0}.
\end{align}
Here, $p$ parametrizes the strength of Alice's measurement with $p=0,1$ corresponding to the optimal one-way measurement described above.
First consider outcome $A_0$ that occurs with probability $P(A_0)=(2+p)/4$.  Alice broadcasts her result and the updated ensemble becomes
\[\{(\tfrac{1+p}{2+p},\ket{0 0}); (\tfrac{1}{2(2+p)},\ket{s_++});(\tfrac{1}{2(2+p)},\ket{s_--})\},\]
where $\ket{s_{\pm}}=\frac{1}{\sqrt{2}}[\sqrt{1+p}\ket{0}\pm\sqrt{1-p}\ket{1}]$.  Bob now pretends that Alice has completely eliminated the state $\ket{1}$, i.e. if she had chosen $p=1$.  That is, he performs optimal discrimination measurement for the ensemble $\{(2/3,\op{0}{0});(1/3,\mathbb{I}/2)\}$.  This amounts to measuring in the computational basis $\{B_0=\op{0}{0},B_1=\op{1}{1}\}$.  If he measures $B_1$ then the state $\rho$ will be perfectly identified.  On the other hand, outcome $B_0$ occurs with probability $P(B_0)=\tfrac{3+2p}{2(2+p)}$ and in that case the updated probabilities of Alice's ensemble are 
\begin{align}
P(0|B_0)&=\tfrac{1}{P(B_0)}\cdot\tfrac{1+p}{2+p}\notag\\
P(s_{\pm}|B_0)&=\tfrac{1}{2P(B_0)}\cdot\tfrac{1}{2(2+p)}\notag
\end{align}
so that she is left to distinguish the sub-normalized states
\begin{align}
\Psi&=\tfrac{1}{P(B_0)}\cdot\tfrac{1+p}{2+p}\op{0}{0},&\rho&=\tfrac{1}{P(B_0)}\cdot\tfrac{1}{4(2+p)}(\op{s_+}{s_+}+\op{s_-}{s_-}).
\end{align}
We use Lemma \ref{Lem:pure-mixed} to compute the optimal probability.  Note that $|\ip{s_{\pm}}{0}|^2=\tfrac{1+p}{2}$ and $|\ip{s_+}{s_-}|^2=p$.  We have
\begin{align}
P(B_0)\cdot|(P(0|B_0)-P(s_+|B_0)-P(s_-|B_0))|&=\frac{1+2p}{2(2+p)}
\end{align}
and 
\begin{align}
P(B_0)^2\det\Delta&=\left(\tfrac{1}{4(2+p)}\right)^2(1-p)-\tfrac{1+p}{2+p}\tfrac{1}{4(2+p)}(1-p)=-\frac{1-p}{4(2+p)}\frac{3+4p}{4(2+p)}<0.
\end{align}
Therefore, given outcomes $A_0$ and $B_0$, the optimal error probability is $\tfrac{1}{2}-\tfrac{\sqrt{4+5p}}{4(2+p)P(B_0)}$, and the weighted error for these outcomes is
\begin{align}
\label{Eq:Perr1}
P(A_0)P(B_0)P(Err|A_0B_0)&=\tfrac{2+p}{4}[\tfrac{3+2p}{4(2+p)}-\tfrac{\sqrt{4+5p}}{4(2+p)}]\notag\\
&=1/16[3+2p-\sqrt{4+5p}].
\end{align}

For the $A_1$ outcome of Alice's measurement, Bob again measures in the computational basis.  An analogous calculation gives the probability $P(A_1)P(B_0)P(Err|A_1B_0)=1/16[3-2p-\sqrt{4-3p}]$.  Therefore, the total error probability for the described protocol is
\begin{equation}
\sum_{\lambda=0}^1P(A_\lambda)P(B_0)P(Err|A_\lambda B_0)=1/16(6-\sqrt{4-3p}-\sqrt{4+5p}).
\end{equation}
The error probability is plotted in Fig. \ref{Fig:2-way-LOCC-new}.  For all values of $p$, this protocol performs at least well as the optimal 1-way LOCC measurement.  For comparison, we note that these probabilities are strictly greater than the minimal error probability obtainable by separable operations, which is $\tfrac{1}{8}(3-\sqrt{5})$ \cite{Chitambar-2013b}.

\begin{figure}[h]
\centering
\includegraphics[scale=0.6]{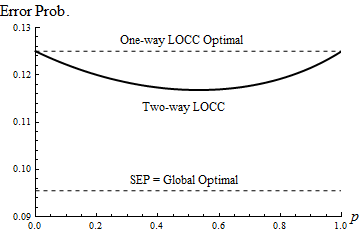}
\caption{\label{Fig:2-way-LOCC-new}
The minimum error probability of a two-way LOCC protocol as a function of $p$.  Note that when $p=0,1$ we obtain $1/8$, which is the smallest error probability using 1-way LOCC; for all other values of $p$ the error probability is strictly smaller.  The lower dotted line gives the optimal separable probability of $\tfrac{1}{8}(3-\sqrt{5})$.} 
\end{figure}

\section{Conclusion}

In conclusion, we have derived a general necessary condition for the optimal discrimination of two states by asymptotic LOCC.  This condition was used to demonstrate a gap between the separable and LOCC distinguishability norms.  We have also computed an explicit two-qubit ensemble whose two-way LOCC discrimination error is strictly less than when classical communication is limited to a single direction.  Theorem \ref{thm:OurKKB} is similar to Proposition 1 in Ref. \cite{Kleinmann-2011a} of Kleinmann \textit{et al.}, and we have stated Theorem \ref{thm:OurKKB} in an analogous manner.  The key difference is that, whereas Ref. \cite{Kleinmann-2011a} provides a necessary condition for just one product operator $\Pi_\lambda$, our theorem holds for a complete set of POVM elements.  It is the use of Carath\'{e}odory's Theorem that makes this possible.  Being able to place conditions on the full POVM is crucial for our particular demonstration of $||\rho-\sigma||_{\text{SEP}}>||\rho-\sigma||_{\text{LOCC}}$.

Could Theorem \ref{thm:OurKKB} also provide sufficient conditions for asymptotic LOCC discrimination between $\rho$ and $\sigma$?  While we currently have no example for when the theorem is insufficient, we suspect that such examples can be found.  Besides demanding that the supports of $\Pi_\lambda$ collectively contain the support of $\rho$, there is no additional condition that relates the individual $\Pi_\lambda$, and the latter seems essential to fully capture the LOCC nature of the measurement process.  On the other hand, it is not clear how Theorem \ref{thm:OurKKB} could be significantly strengthened.  The reason is that any LOCC protocol could always be modified in an infinite number of ways using weak measurement decompositions along the lines of Proposition \ref{Prop:LOCCweak}.  It is hard to see what further constraints could be placed on the $\Pi_\lambda$ (such as orthogonality) from which an exception could not be found among the various modifications.

Finally, we comment on the finding that $||\psi-\rho||_{\text{LOCC}}>||\psi-\rho||_{1-\text{LOCC}}$.  To our knowledge, it was previously unknown whether or not two-way adaptive LOCC protocols actually offer an advantage in the binary discrimination problem.  In fact, the data hiding examples studied in Refs. \cite{Terhal-2001a, DiVincenzo-2002a} have an optimal LOCC discrimination strategy using one-way communication \cite{Matthews-2009a}.  Furthermore, for any two pure states, optimal LOCC discrimination only requires one-way communication, while for any two full-rank mixed states of two qubits, one-way communication is likewise sufficient to achieve LOCC optimality \cite{Chitambar-2013b}.  In the example ensemble of Sect. \ref{Sect:1-wayVs2-way}, the supports of $\psi$ and $\rho$ do not cover the full two-qubit state space.  This suggests that there is some connection between LOCC communication complexity and the ranks of the states being discriminated.  We hope our results here can shed some new light in this direction.

\section*{Acknowledgements}

We thank Debbie Leung, Andreas Winter, and especially Laura Man\v{c}inska for providing helpful comments and corrections during the preparation of this manuscript.


\section*{Appendices}

\appendix

\section{Proof of Proposition \ref{Prop:AlicePOVM}}

We rely on the fact that the ensemble is real with each state being $(\sigma_z\otimes \sigma_z)$-invariant.  Then by an argument similar to the one by Sasaki \textit{et al.}, we can take the $\ket{a_\lambda}$ to be real in the computational basis \cite{Sasaki-1999a}, and moreover, Alice's POVM will consist of pairs
\begin{align}
\op{a_\lambda}{a_\lambda}&=c_\lambda/2(\mathbb{I}+\sin\phi_\lambda\sigma_x+\cos\phi_\lambda\sigma_z)\notag\\
\sigma_z\op{a_\lambda}{a_\lambda}\sigma_z&=c_\lambda/2(\mathbb{I}-\sin\phi_\lambda\sigma_x+\cos\phi_\lambda\sigma_z).\notag
\end{align}
We next divide Alice's outcomes $\lambda$ into two sets: $\lambda\in S_+$ if $\cos\phi_\lambda\geq 0$, and $\lambda\in S_-$ if $\cos\phi_\lambda<0$.  The completion condition demands that $\sum_{\lambda\in S_+} \cos\phi_\lambda c_\lambda+\sum_{\lambda\in S-}\cos\phi_\lambda c_\lambda=0$.  We can fine-grain the POVM elements so that $|S_+|=|S_-|$, and for every $\lambda_+\in S_+$ there exists a corresponding $\lambda_-\in S_-$ so that $\cos\phi_{\lambda_+}c_{\lambda_+}=|\cos\phi_{\lambda_-}|c_{\lambda_-}$.  We see that 
\[\op{a_{\lambda_+}}{a_{\lambda_+}}+\op{a_{\lambda_-}}{a_{\lambda_-}}+\sigma_z\op{a_{\lambda_+}}{a_{\lambda_+}}\sigma_z+\sigma_z\op{a_{\lambda_-}}{a_{\lambda_-}}\sigma_z\propto\mathbb{I}.\]
Thus, Alice's POVM can be decomposed into a collection of sub-POVMs and it suffices to consider optimality among these sub-POVMs.

\section{Optimization of Eq. \mbox{\eqref{Eq:detgeq0}}}

To minimize the RHS of Eq. \eqref{Eq:detgeq0}, we consider the Lagrangian
\begin{equation}
\mathcal{L}(q_1,\phi_0,\phi_1,\tau)=q_1(\frac{\cos\phi_1}{2}+\frac{\sqrt{1+\cos\phi_1}}{\sqrt{2}})-\tau(q_1\cos\phi_1+(1-q_1)\cos\phi_0).
\end{equation}
Exterma points satisfy the equations
\begin{align}
\frac{\partial\mathcal{L}}{\partial\phi_0}&=\tau(1-q_1)\sin\phi_0=0\notag\\
\frac{\partial\mathcal{L}}{\partial\phi_1}&=-\sin\phi_1q_1[\frac{1}{2}+\frac{1}{2\sqrt{2}\sqrt{1+\cos\phi_1}}-\tau ]=0\notag\\
\frac{\partial\mathcal{L}}{\partial q_1}&=\frac{\cos\phi_1}{2}+\frac{\sqrt{1+\cos\phi_1}}{\sqrt{2}}-\tau(\cos\phi_1-\cos\phi_0)=0\notag\\
\frac{\partial\mathcal{L}}{\partial \tau}&=q_1\cos\phi_1+(1-q_1)\cos\phi_0=0.
\end{align}
We first note that if $\tau=0$ then by the third equation, $\frac{\cos\phi_1}{2}+\frac{\sqrt{1+\cos\phi_1}}{\sqrt{2}}=0$ which corresponds to an error probability of $1/2$.  So suppose $\tau\not=0$.  Then the first equation requires $(1-q_1)\sin\phi_0=0$.  If $q_1=1$, then the fourth equation requires $\cos\phi_1=0$, which corresponds to an error probability $\frac{1}{2}\left(1-\frac{1}{\sqrt{2}}\right)$.  On the other hand, if $\sin\phi_0=0$, then we must have $\cos\phi_0=-1$ (recall that $\cos\phi_0<0$ since $\det\Delta>0$ for the $\lambda=0$ branch).  So assume that $\cos\phi_0=-1$.  We turn to the second equation and first consider when $\sin\phi_1q_1=0$.  This requires $\cos\phi_1=1$ and $q_0=q_1=1/2$, which corresponds to Alice measuring in the computational basis.  In this case, we find that the error probability is given by $\frac{1}{2}\left(1-\frac{3}{4}\right)=1/8$.  On the hand if $\sin\phi_1q_1\not=0$ in the second equation, then $\tau=\frac{1}{2}+\frac{1}{2\sqrt{2}\sqrt{1+\cos\phi_1}}$.  However, solving for $\tau$ in the third equation (with $\cos\phi_0=-1$) gives $\tau= \frac{\cos\phi_1}{2(1+\cos\phi_1)}+\frac{1}{\sqrt{2}\sqrt{1+\cos\phi}}$.  Thus,
\[\frac{1}{2}=\frac{\cos\phi_1}{2(1+\cos\phi_1)}+\frac{1}{2\sqrt{2}\sqrt{1+\cos\phi_1}}.\]
This leads to the equation $\sqrt{2}=\sqrt{1+\cos\phi_1}$, which means $\cos\phi_1=1$, corresponding again to measurement in the computational basis.

We have exhausted all possible extreme points.  Alice measuring in the $\{\ket{0},\ket{1}\}$ basis generates a global minimum of $1/8$ which beats the error probability of $\frac{1}{2}\left(1-\frac{1}{\sqrt{2}}\right)$ generated by measuring in the $\{\ket{+},\ket{-}\}$ basis.  Hence the one-way optimal LOCC probability is $1/8$.

\bibliographystyle{alphaurl}
\bibliography{QuantumBib}

\end{document}